\documentclass[11pt]{article}
\usepackage[margin=1in]{geometry}
\usepackage[final]{microtype}
\usepackage{epsfig}
\usepackage{graphics}
\usepackage{latexsym}
\usepackage{amsmath}
\usepackage{amsfonts}
\usepackage{amssymb}
\usepackage{mathrsfs}
\usepackage[dvipsnames]{xcolor}
\usepackage{amsthm}
\usepackage{xspace}
\usepackage{epstopdf}
\usepackage{float}
\usepackage{hyperref}
\usepackage{pgfplots}
\pgfplotsset{compat=1.18}
\usepackage{thmtools}
\usepackage{thm-restate}
\usepackage{longtable}

\numberwithin{equation}{section}
\usepackage[ruled,vlined]{algorithm2e}
\SetArgSty{textrm}

\usepackage[english]{babel}
\usepackage[nottoc]{tocbibind}

\usepackage{caption}
\usepackage{subcaption}

\usepackage{pifont}
\usepackage{booktabs}
\usepackage{multirow}

\theoremstyle{plain}     
\newtheorem{theorem}{Theorem}

\newtheorem{lemma}{Lemma}

\newtheorem{mechanism}{Mechanism}
\theoremstyle{definition} 
\newtheorem{definition}{Definition}

\theoremstyle{remark} 

\makeatletter
\@addtoreset{equation}{section}
\def\section{\@startsection {section}{1}{\z@}{-3.5ex plus -1ex minus
 -.2ex}{2.3ex plus .2ex}{\large\bf}}
\makeatother

\def\bfm#1{\mbox{\boldmath$#1$}}

\def\0{\bfm 0}

\DeclareMathAlphabet{\mathpzc}{OT1}{pzc}{m}{it}

\newcounter{my}

\newcounter{my2}

\newcounter{my3}

\newcounter{my4}

\newcounter{my5}

\newcounter{my6}

\allowdisplaybreaks 

\begin{document}

\title{Improved Metric Distortion Bounds for Deterministic Weighted-Tournament Voting Rules}
\author{Jianan Lin$^{1}$\quad Chenhao Wang$^{2,3}$\\[0.75em]
$1$ Rensselaer Polytechnic Institute\\
$2$ Beijing Normal University-Zhuhai\\
$3$ Beijing Normal-Hong Kong Baptist University
}
\date{}
\maketitle

\begin{abstract}
We study deterministic weighted-tournament voting rules, also known as C2 rules, which use only the fraction of voters who prefer one candidate to another for each pair of candidates.
Voters and candidates lie in a common but unknown metric space, and voters rank candidates by distance.
The goal is to select a candidate minimizing the sum of distances to the voters, and distortion measures the worst-case approximation ratio relative to an optimal candidate.
Charikar et al.~[EC 2025] introduced the Unblanketed Set rule with distortion at most $3.9312$ and proved a lower bound of $3.1128$ for deterministic C2 rules.
We extend their rule to the Path-Unblanketed Set rule, which runs in polynomial time and achieves distortion at most $1+2\sqrt{2}\approx3.8284$.
The rule uses paths between candidates to combine pairwise comparisons without increasing the approximation loss with the path length.
With computer assistance, we prove a better upper bound of $3.3346$ for five or six candidates and a general lower bound of $3.1828$ using a construction with 17 candidates.
\end{abstract}

\section{Introduction}\label{sec:introduction}

Voters can often rank the candidates more easily than they can assign a numerical utility to each candidate.
Voting rules must then choose a winner without knowing the strength of voters' preferences.
Procaccia and Rosenschein~\cite{procaccia2006distortion} introduced \emph{distortion} to measure the welfare loss caused by using rankings instead of numerical utilities.
In the metric model of Anshelevich et al.~\cite{anshelevich2018approximating}, voters and candidates lie in a common metric space, and each voter ranks candidates by distance.
The \emph{social cost} of a candidate is the sum of its distances to the voters.
Distortion is the worst-case ratio between the social cost of the selected candidate and the minimum social cost, over all metrics consistent with the reported rankings \cite{goel2017metric,anshelevich2021distortion}.
For deterministic rules, which select a candidate without randomization, the optimal distortion using voters' full rankings is $3$ \cite{gkatzelis2020resolving,kizilkaya2022plurality}.

Pairwise comparisons are the basis of classical voting rules such as \emph{Copeland}, which selects a candidate who defeats the most other candidates in pairwise majority votes \cite{laslier1997tournament,brandt2016tournament}.
We study rules that use the fraction $w_{ab}$ of voters who prefer $a$ to $b$ for each pair of candidates $a,b$.
These fractions form a \emph{weighted tournament}, and rules that use only this information are called \emph{weighted-tournament rules}, or \emph{C2 rules} \cite{munagala2019improved,charikar2025metric}.
This information does not determine the voters' individual rankings.
Nevertheless, randomized C2 rules, which select a winner according to a probability distribution, achieve the same distortion of $3$ as deterministic rules using full rankings \cite{charikar2024breaking}.
The optimal distortion for deterministic C2 rules remains open.

Anshelevich et al.~\cite{anshelevich2018approximating} proved that Copeland has distortion $5$ and that every deterministic rule using only rankings has distortion at least $3$.
Munagala and Wang~\cite{munagala2019improved} improved the upper bound to $2+\sqrt{5}\approx4.2361$ using the \emph{weighted uncovered set}.
A candidate belongs to this set if, against every other candidate, it has enough voter support either in a direct comparison or in two comparisons through an intermediate candidate.
More recently, Charikar et al.~\cite{charikar2025metric} introduced the \emph{Unblanketed Set} rule, which uses another candidate, called a \emph{proxy}, to bound a candidate's social cost relative to that of every other candidate.
Their rule allows more general comparisons involving a few candidates and achieves distortion at most $3.9312$.
They also proved a lower bound of $3.1128$, even with only five candidates, showing that deterministic C2 rules cannot match the optimal guarantee available from full rankings.
These results leave the optimal distortion for deterministic C2 rules between $3.1128$ and $3.9312$.

We build on the Unblanketed Set rule to narrow this gap.
Its comparisons involve only a small number of intermediate candidates.
We consider paths that allow these comparisons to involve more candidates.
To obtain a better distortion bound, we need to use the comparisons along a path without incurring an additional loss at each step.

\subsection{Our Contributions}

We introduce the \emph{Path-Unblanketed Set rule}, a deterministic C2 rule that extends the Unblanketed Set rule of Charikar et al.~\cite{charikar2025metric} by allowing comparisons along paths.
For any number of candidates, our rule runs in polynomial time and achieves distortion at most $1+2\sqrt{2}\approx3.8284$, improving the previous upper bound of $3.9312$.

The analysis starts from a sufficient condition of Charikar et al.~\cite{charikar2025metric} for comparing the social costs of two candidates.
Their condition uses a proxy and requires an inequality for every partition separating the candidate under consideration from that proxy.
We construct a directed graph in which a path from the candidate to its proxy guarantees these inequalities.
For each partition, an edge crossing it gives the required inequality.
The bound therefore does not grow with the path length.
It remains to prove that some candidate satisfies this condition against every other candidate.
If no candidate satisfies the condition, each candidate fails the comparison with some other candidate.
Drawing an arc for each failed comparison gives a directed graph containing a cycle.
The absence of paths to suitable proxies gives inequalities around the cycle, which we show cannot all hold (Section~\ref{sec:upper-bound}).

For five or six candidates, we prove with computer assistance that the same rule achieves distortion at most $1+2\rho\approx3.3346$, where $\rho>1$ is the real root of $\rho^5=\rho+1$ (Section~\ref{sec:local-bounds}).
We also improve the lower bound for deterministic C2 rules from $3.1128$ to $3.1828$.
Following the approach of Charikar et al.~\cite{charikar2025metric}, we construct a weighted tournament with 17 candidates and use computer assistance to prove the bound (Section~\ref{sec:lower-bound}).
For each candidate, we give a distribution over rankings and a consistent metric in which its social cost is at least $3.1828$ times the minimum social cost.
All these distributions induce the same weighted tournament, so a deterministic C2 rule cannot distinguish them and must choose the same winner in each case.

\subsection{Additional Related Work}

Beyond the metric model, Boutilier et al.~\cite{boutilier2015optimal} studied how well voting rules can maximize total utility when voters report only rankings.
Other work in the metric model has studied fairness and the distortion of common voting rules \cite{goel2017metric,skowron2017social}.
We refer to Anshelevich et al.~\cite{anshelevich2021distortion} for a survey of this literature.

For randomized rules that observe voters' full rankings, Charikar et al.~\cite{charikar2024breaking} first obtained distortion below $3$, with an upper bound of $2.753$.
Frank~\cite{frank2026improved} and Ye~\cite{ye2026half} independently improved the upper bound to $\frac{5}{2}$, followed by Shah's bound of $2.1441$ \cite{shah2026improving} and Cai et al.'s bound of $2.1248$ \cite{cai2026stable}.
The best known lower bound remains $2.11264$ \cite{charikar2022metric}.
Cai et al.~\cite{cai2026distortion} also studied rules that choose uniformly from a list of constant size and showed that these rules can achieve distortion below $3$.
Jia and Peng~\cite{jia2026robust} subsequently obtained distortion $\frac{5}{2}+\varepsilon$ using a list whose size depends only on $\varepsilon$.
When a randomized rule uses only the weighted tournament, the optimal distortion is $3$ \cite{goel2017metric,charikar2024breaking}.
Allowing aggregate rankings over triples of candidates improves this bound to $\frac{7}{3}$, which Cai et al.~\cite{cai2026stable} proved optimal for that information model.
If the rule knows only which candidate has majority support in each pairwise comparison, without knowing the fraction of voters supporting it, the optimal randomized distortion is $4$ \cite{frank2025metric}.

The weighted uncovered set and the Unblanketed Set discussed above are related to classical methods for choosing candidates from pairwise majority comparisons \cite{miller1980new,moulin1986choosing,dutta1988covering}.
For a broader treatment of these methods, see Laslier~\cite{laslier1997tournament} and Brandt et al.~\cite{brandt2016tournament}.
Other variants limit how much information voters communicate to the rule \cite{kempe2020communication} or allow the rule to return a committee rather than a single winner \cite{banihashem2026bi}.
Kizilkaya et al.~\cite{kizilkaya2026metric} extended metric distortion to rules that output a ranking of the candidates, with costs depending on their positions in the ranking.
Frank and Peters~\cite{frank2026constant} studied approval voting when voters also know each candidate's share of first-choice votes, and gave a protocol with distortion $2+\sqrt{5}$.

\section{Preliminaries}\label{sec:preliminaries}

Let $V$ be a set of $n$ voters, and let $C$ be a set of $m$ candidates.
Each voter $v\in V$ reports a strict total order over $C$.
We write $a\succ_v b$ when voter $v$ prefers candidate $a$ to candidate $b$, and denote the resulting preference profile by $\sigma=(\succ_v)_{v\in V}$.
A deterministic social choice rule selects a winning candidate based on the preference information available to it.

The \emph{weighted tournament} induced by $\sigma$ is the matrix $W(\sigma)=(w_{ab})_{a,b\in C}$.
For distinct candidates $a,b$, the entry
\[
    w_{ab}=\frac{|\{v\in V:a\succ_v b\}|}{n}
\]
is the fraction of voters who prefer $a$ to $b$.
Thus, $w_{ab}+w_{ba}=1$.
We use the convention $w_{aa}=\frac{1}{2}$.
A deterministic \emph{weighted-tournament rule}, also called a deterministic \emph{C2 rule}, selects a winning candidate using only $W(\sigma)$ \cite{munagala2019improved,charikar2025metric}.
Consequently, if two preference profiles induce the same weighted tournament, a C2 rule must select the same winner for both profiles.

The voters and candidates lie in a common pseudo-metric space $(X,d)$, where $X=V\cup C$, and the rule does not observe $d$.
Distances are nonnegative and satisfy $d(x,x)=0$, $d(x,y)=d(y,x)$, and $d(x,z)\le d(x,y)+d(y,z)$ for all $x,y,z\in X$.
Unlike a metric, a pseudo-metric allows distinct voters or candidates to be at distance zero.
The pseudo-metric $d$ is \emph{consistent} with a preference profile $\sigma$ if $a\succ_v b$ implies $d(v,a)\le d(v,b)$ for every voter $v\in V$ and every pair $a,b\in C$.
Candidates at the same distance from a voter may appear in either order in that voter's ranking.
Let $\mathcal D(\sigma)$ denote the set of pseudo-metrics consistent with $\sigma$.

For a candidate $c\in C$, its \emph{social cost} under $d$ is
\[
    \mathrm{SC}(c,d)=\sum_{v\in V}d(v,c).
\]
An optimal candidate minimizes $\mathrm{SC}(c,d)$ over $c\in C$.
The \emph{distortion} of a deterministic C2 rule $f$, denoted by $\operatorname{dist}(f)$, is the smallest $\alpha\in[1,\infty]$ such that
\begin{equation}
    \mathrm{SC}\bigl(f(W(\sigma)),d\bigr)
    \le \alpha\min_{c\in C}\mathrm{SC}(c,d)
    \label{eq:distortion}
\end{equation}
for every preference profile $\sigma$ and every $d\in\mathcal D(\sigma)$.
For a fixed number $m$ of candidates, let $D_m^{\mathrm{C2}}$ denote the infimum distortion over all deterministic C2 rules on $m$ candidates.
For a rule defined for every number of candidates, its overall distortion is the supremum of its distortion over all integers $m\ge 2$.

\section{Upper Bound for Deterministic C2 Rules}\label{sec:upper-bound}

We begin by recalling the Unblanketed Set rule of Charikar et al.~\cite{charikar2025metric}.
Fix parameters $\alpha\ge\beta>\frac12$.
A candidate $j$ is \emph{$(\alpha,\beta)$-unblanketed} if, for every $i\ne j$, there is a candidate $k$ such that either $k=i$ or $w_{ki}\ge\alpha$, and at least one of the following conditions holds:
\begin{enumerate}
    \item $w_{kj}\le\beta$;
    \item $w_{kj}\le\alpha$ and there is a candidate $\ell\notin\{k,j\}$ such that $w_{k\ell}\le\beta\le w_{j\ell}$.
\end{enumerate}

\begin{mechanism}[Unblanketed Set Rule \cite{charikar2025metric}]\label{mech:unblanketed}
Given a weighted tournament and parameters $\alpha\ge\beta>\frac12$, select an $(\alpha,\beta)$-unblanketed candidate using fixed tie-breaking.
\end{mechanism}

Charikar et al. prove that an $(\alpha,\beta)$-unblanketed candidate always exists and, with their choice of parameters, that the rule has distortion at most $3.9312$.
The first condition compares $j$ and $k$ directly, while the second uses one additional candidate $\ell$.
Our rule allows longer paths in a directed graph, with the same fixed candidates $j$ and $k$ used to define every edge.

\subsection{The Path-Unblanketed Set Rule}

Fix a parameter $\lambda\ge1$.
Following Charikar et al., we use a candidate $k$ as a \emph{proxy} for $i$ when comparing the social costs of $j$ and $i$.
When $k\ne i$, we require $w_{ij}\le\lambda w_{ki}$.
Thus, greater voter support for $i$ over $j$ requires greater support for $k$ over $i$.

For distinct candidates $j,k\in C$, define the directed graph $G_\lambda(j,k)$ on vertex set $C$ by including an arc $a\to b$ if and only if
\[
    w_{bj}\le\lambda w_{ak}.
\]
The inequality compares voter support for $b$ over $j$ with voter support for $a$ over $k$.
A directed path from $j$ to $k$ in this graph has the form
\[
    j=z_0\to z_1\to\cdots\to z_h=k,
\]
where every step satisfies
\[
    w_{z_{t+1},j}\le\lambda w_{z_t,k}
    \qquad (t=0,\ldots,h-1).
\]
Paths allow us to apply the social-cost comparison of Charikar et al.~\cite[Corollary~4.3]{charikar2025metric}, which we recall below.
Their condition requires an inequality for every partition separating $j$ from $k$.
A path from $j$ to $k$ crosses every such partition, and the inequality defining a crossing edge is sufficient for that partition.
Different partitions may use different edges, so the length of the path does not affect the bound.

\begin{definition}[Admissibility and path-unblanketed candidates]
For distinct candidates $i,j\in C$, a candidate $k$ is \emph{$\lambda$-admissible for $(i,j)$} if $k\ne j$ and either $k=i$ or $w_{ij}\le\lambda w_{ki}$.
A candidate $j$ is \emph{$\lambda$-path-unblanketed} if, for every $i\in C\setminus\{j\}$, there is a candidate $k$ that is $\lambda$-admissible for $(i,j)$ and reachable from $j$ in $G_\lambda(j,k)$.
\end{definition}

The candidate $i$ itself is always an admissible proxy for $(i,j)$.
If $w_{ij}\le\lambda w_{ji}$, the direct arc $j\to i$ in $G_\lambda(j,i)$ already gives a path to this proxy.

\begin{mechanism}[$\lambda$-Path-Unblanketed Set Rule]\label{mech:path-unblanketed}
Given a weighted tournament and a parameter $\lambda\ge1$, select a $\lambda$-path-unblanketed candidate using fixed tie-breaking.
\end{mechanism}

The rule can be implemented directly from the definition.
Inspect the candidates $j$ in the fixed tie-breaking order.
For every $i\ne j$, enumerate the possible candidates $k$, check admissibility, construct $G_\lambda(j,k)$, and check whether $k$ is reachable from $j$.
Select the first $j$ for which a suitable $k$ is found for every $i\ne j$.
There are polynomially many graphs and each reachability test takes polynomial time.
We prove that such a candidate always exists when $\lambda=\sqrt{2}$.

\begin{theorem}\label{thm:main-upper-bound}
Every weighted tournament contains a $\sqrt{2}$-path-unblanketed candidate.
The $\sqrt{2}$-Path-Unblanketed Set Rule is a polynomial-time deterministic C2 rule with distortion at most $1+2\sqrt{2}\approx3.8284$.
\end{theorem}

This bound improves the previous upper bound of $3.9312$ for deterministic C2 rules~\cite{charikar2025metric}.
We begin the proof with the following comparison of social costs, due to Charikar et al.~\cite{charikar2025metric}, which holds for any candidate $i$.
\begin{lemma}[\cite{charikar2025metric}, Corollary~4.3]
$\mathrm{SC}(j,d)\le(1+2\lambda)\mathrm{SC}(i,d)$ if there exists a candidate $k$ such that either $k=i$ or $w_{ij}\le \lambda w_{ki}$, and for all partitions of the candidates $I\sqcup J=C$ such that $i,k\in I$ and $j\in J$,
     \begin{align}
    \min_{a\in I} w_{aj}
    &\le
    \lambda\max_{b\in J}\{w_{bi},w_{bk}\}.
    \label{eq:path-cut-certificate}
\end{align}
\end{lemma}

The following lemma shows that a path to an admissible proxy satisfies these conditions.

\begin{lemma}\label{lem:path-certificate}
Fix $\lambda\ge 1$ and distinct candidates $i,j\in C$.
Suppose that $d\in\mathcal D(\sigma)$ is a consistent metric.
If some candidate $k$ is $\lambda$-admissible for $(i,j)$ and reachable from $j$ in $G_\lambda(j,k)$, then $\mathrm{SC}(j,d)\le(1+2\lambda)\mathrm{SC}(i,d)$.
\end{lemma}

\begin{proof}
Admissibility gives the first condition of the preceding lemma.
It remains to verify \eqref{eq:path-cut-certificate} for every partition $I\sqcup J=C$ with $i,k\in I$ and $j\in J$.
Because $k$ is reachable from $j$ in $G_\lambda(j,k)$, there is a directed path $j=z_0\to z_1\to\cdots\to z_q=k$.
By the definition of $G_\lambda(j,k)$, we have $w_{z_{\ell+1},j}\le\lambda w_{z_\ell,k}$ for every $\ell\in\{0,\ldots,q-1\}$.
Fix an arbitrary partition $I\sqcup J=C$ with $i,k\in I$ and $j\in J$.
The path starts at $z_0=j\in J$ and ends at $z_q=k\in I$.
It must therefore cross the cut from $J$ to $I$ at least once.
Let $r$ be the smallest index for which $z_r\in I$, so $r\ge1$ and $z_{r-1}\in J$.
We have
\begin{align*}
    \min_{a\in I} w_{aj}
    &\le w_{z_r,j}
    \le \lambda w_{z_{r-1},k}
    \le \lambda\max_{b\in J}\{w_{bi},w_{bk}\},
\end{align*}
where the first and last inequalities follow from $z_r\in I$ and $z_{r-1}\in J$, and the middle inequality follows from the arc $z_{r-1}\to z_r$.
Since the partition was arbitrary, \eqref{eq:path-cut-certificate} holds for every required partition.
For each partition we use only one edge, so the bound contains a single factor of $\lambda$, regardless of the path length.
The preceding lemma therefore implies $\mathrm{SC}(j,d)\le(1+2\lambda)\mathrm{SC}(i,d)$.
\end{proof}

Taking $i$ to be optimal in Lemma~\ref{lem:path-certificate} gives distortion at most $1+2\lambda$ for any $\lambda$-path-unblanketed candidate.
To prove Theorem~\ref{thm:main-upper-bound}, it remains to establish the existence of such a candidate for $\lambda=\sqrt{2}$.

\subsection{Bad Cycles and Reachability}

For distinct candidates $i,j\in C$, call $(i,j)$ a \emph{$\lambda$-bad pair} if no $\lambda$-admissible candidate $k$ for $(i,j)$ is reachable from $j$ in $G_\lambda(j,k)$.
A bad pair means that we cannot use Lemma~\ref{lem:path-certificate} to compare $j$ with $i$.
If no candidate is $\lambda$-path-unblanketed, every $j$ belongs to a bad pair $(i,j)$.
As in the existence proof of Charikar et al.~\cite[Lemma~4.8]{charikar2025metric}, we first obtain a cycle of such pairs and then construct a second cycle.
We will use the sets of reachable candidates on the second cycle to derive a contradiction.

\begin{lemma}\label{lem:bad-cycle}
If no candidate is $\lambda$-path-unblanketed, then the directed graph containing an arc $i\to j$ for every $\lambda$-bad pair $(i,j)$ contains a directed cycle.
Moreover, every bad pair $(i,j)$ satisfies $w_{ij}>\lambda w_{ji}$.
\end{lemma}

\begin{proof}
Suppose that no candidate is $\lambda$-path-unblanketed.
For every candidate $j$, choose $i(j)\ne j$ such that $(i(j),j)$ is a $\lambda$-bad pair.
This gives an arc $i(j)\to j$.
Thus, every vertex has an incoming arc.

Starting from any candidate $j_0$, recursively define $j_{q+1}=i(j_q)$.
For every $q\ge 0$, the graph contains the arc $j_{q+1}\to j_q$.
Because $C$ is finite, there is a smallest index $q>0$ for which $j_q=j_p$ for some $p<q$.
By the choice of $q$, the vertices $j_p,j_{p+1},\ldots,j_{q-1}$ are distinct, and the arcs $j_q\to j_{q-1}\to\cdots\to j_p=j_q$ form a directed cycle.

It remains to prove the inequality for a bad pair $(i,j)$.
The candidate $k=i$ is automatically $\lambda$-admissible for $(i,j)$.
Since $(i,j)$ is bad, there is no path from $j$ to $i$ in $G_\lambda(j,i)$.
In particular, the direct arc $j\to i$, which exists exactly when $w_{ij}\le\lambda w_{ji}$, is absent.
Therefore, $w_{ij}>\lambda w_{ji}$.
\end{proof}

Fix a directed cycle of bad pairs.
For a candidate $k$ on this cycle, $k$ is an admissible proxy for the edge leaving $k$, but not for the edge entering $k$.
For each candidate $k$ on this cycle, traverse the cycle from the edge leaving $k$ to the edge entering $k$, and let $F(k)$ be the head of the last bad edge for which $k$ is admissible.
We form a new directed graph with the arcs $k\to F(k)$.

\begin{lemma}\label{lem:last-admissible}
The map $F$ is well defined and satisfies $F(k)\ne k$.
Moreover, $k$ is not reachable from $F(k)$ in $G_\lambda(F(k),k)$, and
\begin{align}
    \frac{\lambda}{1+\lambda}
    &<
    w_{k,F(k)}
    <
    \frac{1}{\lambda}.
    \label{eq:last-admissible-bounds}
\end{align}
Consequently, the graph with arcs $k\to F(k)$ contains a directed cycle of length at least two.
\end{lemma}

\begin{proof}
Write the fixed bad-pair cycle as $c_0\to c_1\to\cdots\to c_{q-1}\to c_0$, with indices modulo $q$.
Fix a vertex $k=c_s$.
We call $k$ admissible for an edge $x\to y$ if it is admissible for the bad pair $(x,y)$.
The proxy $k$ is admissible for the edge $(c_s,c_{s+1})$ leaving $k$, because it equals the first candidate of this bad pair.
It is not admissible for the edge $(c_{s-1},c_s)$ entering $k$, because a proxy for this pair is required to differ from its second candidate $c_s=k$.
The last edge for which $k$ is admissible therefore exists, and its head $F(k)$ differs from $k$.

Denote the last admissible edge by $i\to j$ and the next edge on the cycle by $j\to h$.
Thus, $F(k)=j$.
By the choice of the last admissible edge, $k$ is not admissible for the following edge $(j,h)$.
Because $(i,j)$ is a bad pair and $k$ is admissible for it, $k$ is not reachable from $j$ in $G_\lambda(j,k)$.
In particular, the arc $j\to k$ is absent, so $w_{kj}>\lambda w_{jk}=\lambda(1-w_{kj})$.
Rearranging yields $w_{kj}>\frac{\lambda}{1+\lambda}$.

We next derive the upper bound on $w_{kj}$ from the edge $j\to h$ immediately following the selected edge.
If $h\ne k$, then $k$ is distinct from both $j$ and $h$, so its nonadmissibility for $(j,h)$ gives $w_{jh}>\lambda w_{kj}$.
If $h=k$, Lemma~\ref{lem:bad-cycle} applied to the bad pair $(j,k)$ gives $w_{jk}>\lambda w_{kj}$.
In either case, $1\ge w_{jh}>\lambda w_{kj}$.
It follows that $w_{kj}<\frac{1}{\lambda}$, completing \eqref{eq:last-admissible-bounds}.

Finally, $F$ maps the finite vertex set of the original cycle to itself.
Repeated application of $F$ must eventually revisit a vertex, giving a directed cycle.
Since $F(k)\ne k$ for every $k$, the cycle has length at least two.
\end{proof}

On this new cycle, each arc $k\to F(k)$ satisfies \eqref{eq:last-admissible-bounds}, and $k$ is unreachable from $F(k)$ in $G_\lambda(F(k),k)$.
We next express the absence of these paths as inequalities between entries of the weighted tournament.
Fix a directed $F$-cycle, identify its vertices with $\mathbb Z_r=\{0,\ldots,r-1\}$, and relabel it so that $F(t)=t+1$ modulo $r$.
Let $\widehat R_t$ be the set of all candidates reachable from $t+1$ in $G_\lambda(t+1,t)$, and let $R_t=\widehat R_t\cap\mathbb Z_r$.
Reachability is computed in the full graph, so these paths may pass through candidates outside the cycle.
An edge from a reachable vertex to an unreachable vertex cannot exist.
The definition of $G_\lambda(t+1,t)$ therefore gives the following inequalities.

\begin{lemma}\label{lem:reachability-cut}
For every $t\in\mathbb Z_r$, we have $t+1\in R_t$ and $t\notin R_t$.
Moreover, if $a\in R_t$ and $b\in\mathbb Z_r\setminus R_t$, then $w_{b,t+1}>\lambda w_{a,t}$.
\end{lemma}

\begin{proof}
The vertex $t+1$ is reachable from itself, so $t+1\in R_t$.
Lemma~\ref{lem:last-admissible} gives $t\notin R_t$.
Now fix $a\in R_t$ and $b\in\mathbb Z_r\setminus R_t$.
Both are cycle vertices, so $a$ is reachable from $t+1$ in the full graph and $b$ is not.
If $w_{b,t+1}\le\lambda w_{a,t}$, then the definition of $G_\lambda(t+1,t)$ would include the arc $a\to b$.
Concatenating a path from $t+1$ to $a$ with this arc would make $b$ reachable, a contradiction.
Therefore, $w_{b,t+1}>\lambda w_{a,t}$.
\end{proof}

We now compare the reachable sets at consecutive indices.
With all indices taken modulo $r$, define $y_t=w_{t,t+1}$ as the pairwise weight on the $F$-cycle arc $t\to t+1$, and $p_t=\max_{a\in R_t}w_{a,t}$ as the largest weight against the proxy $t$ among the reachable cycle vertices.
Let $E_t=R_t\setminus R_{t-1}$ be the vertices that belong to $R_t$ but not to $R_{t-1}$.
The next lemma shows that $p_t>\lambda p_{t-1}$ whenever $E_t$ is nonempty.

\begin{lemma}\label{lem:entrant-growth}
For every $t\in\mathbb Z_r$,
\begin{align}
    \frac{\lambda}{1+\lambda}
    &<
    y_t
    <
    \frac{1}{\lambda},
    \qquad
    1-\frac{1}{\lambda}
    <
    p_t
    <
    \frac{y_t}{\lambda}
    <
    \frac{1}{\lambda^2}.
    \label{eq:entrant-growth-bounds}
\end{align}
If $E_t\ne\varnothing$, then $p_t>\lambda p_{t-1}$.
\end{lemma}

\begin{proof}
Since $F(t)=t+1$, the first two inequalities in \eqref{eq:entrant-growth-bounds} are exactly \eqref{eq:last-admissible-bounds} with $k=t$.
By Lemma~\ref{lem:reachability-cut}, the candidate $t+1$ belongs to $R_t$.
In particular, $R_t$ is nonempty and the maximum defining $p_t$ exists.
Since $w_{t+1,t}=1-y_t$, we have $p_t\ge w_{t+1,t}=1-y_t>1-\frac{1}{\lambda}$.

The same lemma gives $t\notin R_t$.
For every $a\in R_t$, substituting $b=t$ into the cut inequality gives $y_t=w_{t,t+1}>\lambda w_{a,t}$.
Taking the maximum over $a\in R_t$ gives $p_t<\frac{y_t}{\lambda}$.
Together with $y_t<\frac{1}{\lambda}$, this proves all inequalities in \eqref{eq:entrant-growth-bounds}.

Suppose now that $E_t\ne\varnothing$, and choose $a\in E_t$.
Then $a\in R_t$ but $a\notin R_{t-1}$.
Choose $u\in R_{t-1}$ with $w_{u,t-1}=p_{t-1}$.
Such a vertex exists because $t\in R_{t-1}$ by Lemma~\ref{lem:reachability-cut}.
At index $t-1$, $u$ is reachable and $a$ is not.
Applying Lemma~\ref{lem:reachability-cut} at index $t-1$ gives $w_{a,t}>\lambda w_{u,t-1}=\lambda p_{t-1}$.
Since $a\in R_t$, we obtain $p_t\ge w_{a,t}>\lambda p_{t-1}$.
\end{proof}

If $E_t\ne\varnothing$ at every index, repeated application of Lemma~\ref{lem:entrant-growth} around the cycle would give $p_t>\lambda^r p_t$, a contradiction.
The next lemma handles the possibility that some of the sets $E_t$ are empty.

\subsection{Cyclic Expansion}

The conditions $t+1\in R_t$ and $t\notin R_t$ prevent all the sets $R_t$ from being equal.
At an index where $E_t$ is empty, however, we have $R_t\subseteq R_{t-1}$.
We show that these restrictions force at least one of the following three configurations, each of which will contradict the preceding bounds when $\lambda=\sqrt{2}$.

\begin{lemma}\label{lem:cyclic-expansion}
Let $R_t\subseteq\mathbb Z_r$ satisfy $t+1\in R_t$ and $t\notin R_t$ for every $t\in\mathbb Z_r$, and let $E_t=R_t\setminus R_{t-1}$.
At least one of the following configurations must occur:
\begin{enumerate}
    \item There exist distinct $s,t$ such that $s\in R_t$ and $t\in R_s$.
    \item There exists $t$ such that $E_t\ne\varnothing$ and $E_{t+1}\ne\varnothing$.
    \item There exist $t,b$ such that $E_t\ne\varnothing$, $b\notin R_t$, and $t+1\notin R_{b-1}$.
\end{enumerate}
\end{lemma}

\begin{proof}
Suppose for contradiction that none of the three configurations occurs.
Let $U=\{t\in\mathbb Z_r:E_t\ne\varnothing\}$.
We first show that $U$ is nonempty.
If $U$ were empty, then $R_t\subseteq R_{t-1}$ for every $t$.
Following these inclusions once around the cycle returns to the starting set, so every inclusion must be an equality and all the sets $R_t$ would coincide.
For any $t$, however, the assumptions give $t+1\in R_t$ and $t+1\notin R_{t+1}$, contradicting that equality.

List the elements of $U$ in cyclic order as $u_0,u_1,\ldots,u_{s-1}$.
Subscripts on the sequence $(u_a)$ are taken modulo $s$, while arithmetic on the values $u_a\in\mathbb Z_r$ is taken modulo $r$.
Intervals between elements of $U$ are read in forward cyclic order, possibly passing through $0$.
Since configuration (2) does not occur, at least one index outside $U$ lies strictly between $u_a$ and $u_{a+1}$ for every $a$.
The definition of $E_v$ gives
\begin{align}
    R_v&\subseteq R_{v-1} \qquad\text{for every }v\notin U.
    \label{eq:nonentrant-inclusion}
\end{align}

The assumption $t+1\in R_t$ gives $u_{a+1}\in R_{u_{a+1}-1}$.
Every index encountered after $u_a$ and before $u_{a+1}$ lies outside $U$, so the sets can only shrink while moving forward through this interval.
Following these inclusions backwards gives
\begin{align}
    u_{a+1}&\in R_{u_a}
    \qquad\text{for every }a.
    \label{eq:cyclic-base-containment}
\end{align}
If $s=1$, \eqref{eq:cyclic-base-containment} says $u_0\in R_{u_0}$, contradicting $u_0\notin R_{u_0}$.
Hence, it remains only to consider $s\ge 2$.

Since configuration (3) does not occur, we have
\begin{align}
    t\in U \text{ and } c\notin R_t
    &\quad\Longrightarrow\quad
    t+1\in R_{c-1}.
    \label{eq:no-crossing-implication}
\end{align}
We will use this implication to show that $u_b\in R_{u_a}$ implies $u_{b+1}\in R_{u_{a-1}}$.
Fix $a,b$ such that $u_b\in R_{u_a}$.
Since $u_a\notin R_{u_a}$, we have $u_b\ne u_a$.
Since configuration (1) does not occur, $u_a\notin R_{u_b}$.
Applying \eqref{eq:no-crossing-implication} with $t=u_b$ and $c=u_a$ gives $u_b+1\in R_{u_a-1}$.

The indices strictly between the consecutive elements $u_{a-1}$ and $u_a$ lie outside $U$.
Repeated use of \eqref{eq:nonentrant-inclusion} therefore gives $R_{u_a-1}\subseteq R_{u_{a-1}+1}$, and hence $u_b+1\in R_{u_{a-1}+1}$.
If $u_b+1=u_{a-1}+1$, this contradicts $t\notin R_t$.
We may therefore assume that these two candidates are distinct.
The absence of configuration (1) then implies $u_{a-1}+1\notin R_{u_b+1}$.

All indices from $u_b+1$ through $u_{b+1}-1$ lie outside $U$.
Consequently, \eqref{eq:nonentrant-inclusion} gives $R_{u_{b+1}-1}\subseteq R_{u_b+1}$, and thus $u_{a-1}+1\notin R_{u_{b+1}-1}$.
If $u_{b+1}\notin R_{u_{a-1}}$, then applying \eqref{eq:no-crossing-implication} with $t=u_{a-1}$ and $c=u_{b+1}$ would instead give $u_{a-1}+1\in R_{u_{b+1}-1}$.
This contradiction proves the expansion rule
\begin{align}
    u_b\in R_{u_a}
    &\quad\Longrightarrow\quad
    u_{b+1}\in R_{u_{a-1}}.
    \label{eq:cyclic-expansion-rule}
\end{align}
Thus, each application replaces $a$ by $a-1$ and $b$ by $b+1$ in the containment $u_b\in R_{u_a}$.

Starting from \eqref{eq:cyclic-base-containment} and applying \eqref{eq:cyclic-expansion-rule} repeatedly, we obtain, for every integer $h\ge 0$,
\begin{align}
    u_{a+1+h}&\in R_{u_{a-h}}.
    \label{eq:cyclic-iterated-expansion}
\end{align}
Suppose first that $s$ is odd, and set $h=\frac{s-1}{2}$ in \eqref{eq:cyclic-iterated-expansion}.
Then $a+1+h$ and $a-h$ are equal modulo $s$, so \eqref{eq:cyclic-iterated-expansion} asserts that some candidate belongs to its own set, contradicting $t\notin R_t$.
Suppose instead that $s$ is even, and set $h=\frac{s-2}{2}$.
Writing $x=a-h$ modulo $s$, equation \eqref{eq:cyclic-iterated-expansion} becomes $u_{x-1}\in R_{u_x}$.
On the other hand, \eqref{eq:cyclic-base-containment} with index $x-1$ gives $u_x\in R_{u_{x-1}}$.
This is configuration (1), again a contradiction.
Both cases give a contradiction, so at least one of the three configurations must occur.
\end{proof}

We now set $\lambda=\sqrt{2}$.
In configurations (1) and (3), the preceding bounds force $w_{ab}+w_{ba}<1$ and $w_{ab}+w_{ba}>1$, respectively, for some distinct candidates $a,b$.
In configuration (2), two consecutive increases force $p_t$ above its upper bound.

\begin{lemma}\label{lem:sqrt2-exclusion}
When $\lambda=\sqrt{2}$, each of the three configurations in Lemma~\ref{lem:cyclic-expansion} is impossible.
\end{lemma}

\begin{proof}
We exclude the three configurations separately.
For configuration (1), distinct candidates $s,t$ satisfy $s\in R_t$ and $t\in R_s$.
Since $s\in R_t$ and $t\notin R_t$, Lemma~\ref{lem:reachability-cut} with $a=s$ and $b=t$ gives $y_t=w_{t,t+1}>\lambda w_{s,t}$.
Interchanging $s$ and $t$ and using the cut at index $s$ similarly gives $y_s>\lambda w_{t,s}$.
Using \eqref{eq:entrant-growth-bounds} in both inequalities yields
\begin{align*}
    w_{s,t}
    <\frac{y_t}{\lambda}
    <\frac{1}{\lambda^2},
    \qquad
    w_{t,s}
    &<\frac{y_s}{\lambda}
    <\frac{1}{\lambda^2}.
\end{align*}
Therefore, $w_{s,t}+w_{t,s}<\frac{2}{\lambda^2}=1$, contradicting $w_{s,t}+w_{t,s}=1$.

For configuration (2), both $E_t$ and $E_{t+1}$ are nonempty.
Lemma~\ref{lem:entrant-growth} gives $p_t>\lambda p_{t-1}$ and $p_{t+1}>\lambda p_t$.
Combining these inequalities with the lower bound on $p_{t-1}$ gives
\begin{align*}
    p_{t+1}
    &>\lambda p_t
    >\lambda^2p_{t-1}
    >\lambda^2\left(1-\frac{1}{\lambda}\right)
    =\lambda^2-\lambda.
\end{align*}
At $\lambda=\sqrt{2}$, we have $\lambda^2-\lambda=2-\sqrt{2}>\frac{1}{2}=\frac{1}{\lambda^2}$.
This contradicts the upper bound $p_{t+1}<\frac{1}{\lambda^2}$ in \eqref{eq:entrant-growth-bounds}.

For configuration (3), there are indices $t,b$ such that $E_t\ne\varnothing$, $b\notin R_t$, and $t+1\notin R_{b-1}$.
Choose $a\in E_t$, so $a\notin R_{t-1}$ and $a\in R_t$.
Since $t\in R_{t-1}$ and $a\notin R_{t-1}$, the cut inequality at index $t-1$ gives $w_{a,t}>\lambda w_{t,t-1}=\lambda(1-y_{t-1})$ by Lemma~\ref{lem:reachability-cut}.
At index $t$, we have $a\in R_t$ and $b\notin R_t$.
Applying Lemma~\ref{lem:reachability-cut} again gives
\begin{align*}
    w_{b,t+1}
    &>\lambda w_{a,t}
    >\lambda^2(1-y_{t-1})
    >\lambda^2-\lambda,
\end{align*}
where the final inequality uses $y_{t-1}<\frac{1}{\lambda}$ from \eqref{eq:entrant-growth-bounds}.

To bound $w_{t+1,b}$, consider the cut at index $b-1$.
We have $b\in R_{b-1}$ and $t+1\notin R_{b-1}$, so $b\ne t+1$.
Lemma~\ref{lem:reachability-cut} therefore yields $w_{t+1,b}>\lambda w_{b,b-1}=\lambda(1-y_{b-1})>\lambda-1$.
But $w_{b,t+1}+w_{t+1,b}=1$.
Adding the last two strict inequalities gives
\begin{align*}
    w_{b,t+1}+w_{t+1,b}
    &>(\lambda^2-\lambda)+(\lambda-1)
    =\lambda^2-1
    =1,
\end{align*}
a contradiction.
All three configurations are therefore impossible when $\lambda=\sqrt{2}$.
\end{proof}

Lemmas~\ref{lem:cyclic-expansion} and~\ref{lem:sqrt2-exclusion} show that the cycle constructed from $F$ cannot exist when $\lambda=\sqrt{2}$.
We can now prove the main theorem.

\begin{proof}[Proof of Theorem~\ref{thm:main-upper-bound}]
Set $\lambda=\sqrt{2}$ and suppose for contradiction that no candidate is $\lambda$-path-unblanketed.
Lemma~\ref{lem:bad-cycle} produces a directed cycle of $\lambda$-bad pairs.
Fix this cycle and define $F(k)$ using the last edge for which $k$ is admissible, as above.
Lemma~\ref{lem:last-admissible} produces a directed cycle of $F$ of length at least two.
Relabel that $F$-cycle by $\mathbb Z_r$ so that $F(t)=t+1$, and construct the sets $R_t$ from the reachability of $t+1$ in $G_\lambda(t+1,t)$.
Lemma~\ref{lem:reachability-cut} gives $t+1\in R_t$ and $t\notin R_t$ for every $t\in \mathbb Z_r$.
The bounds in Lemma~\ref{lem:entrant-growth} also hold for these sets.

The hypotheses of Lemma~\ref{lem:cyclic-expansion} now hold, so the family $(R_t)_{t\in\mathbb Z_r}$ must contain one of the three configurations.
Lemma~\ref{lem:sqrt2-exclusion} shows that every configuration is impossible when $\lambda=\sqrt{2}$.
This contradiction proves that at least one $\sqrt{2}$-path-unblanketed candidate exists.

The rule is therefore well defined.
It uses only $W(\sigma)$ and fixed tie-breaking, so it is a deterministic C2 rule.
The rule can be implemented in polynomial time as follows.
For every ordered pair of distinct candidates $(j,k)$, construct $G_\lambda(j,k)$ and compute the vertices reachable from $j$ by a standard graph search.
There are $O(m^2)$ such graphs, each having $m$ vertices and at most $m^2$ arcs, so all these searches take polynomial time.
Afterwards, for every pair $(i,j)$, scan the candidates $k$ and check admissibility together with the stored reachability result.
Thus, the entire rule runs in polynomial time.

Fix any preference profile $\sigma$, any consistent metric $d\in\mathcal D(\sigma)$, and let $j$ be the candidate returned by the rule.
Let $i$ be a candidate minimizing $\mathrm{SC}(i,d)$.
If $i=j$, the desired social-cost bound is immediate.
Suppose $i\ne j$.
Because $j$ is $\sqrt{2}$-path-unblanketed, there is a $\sqrt{2}$-admissible candidate $k$ for $(i,j)$ that is reachable from $j$ in $G_{\sqrt{2}}(j,k)$.
Lemma~\ref{lem:path-certificate} gives
\begin{align*}
    \mathrm{SC}(j,d)
    &\le (1+2\sqrt{2})\mathrm{SC}(i,d).
\end{align*}
Since this inequality holds for every $\sigma$ and every consistent $d$, definition \eqref{eq:distortion} gives distortion at most $1+2\sqrt{2}$.
\end{proof}

The bound in Theorem~\ref{thm:main-upper-bound} is tight for Mechanism~\ref{mech:path-unblanketed} with fixed tie-breaking, even with two candidates.
Let \(j\) precede \(i\) in the fixed tie-breaking order and let $w_{ij}=\frac{\sqrt{2}}{1+\sqrt{2}}$, $w_{ji}=\frac{1}{1+\sqrt{2}}$.
Both candidates are \(\sqrt{2}\)-path-unblanketed, so the rule selects \(j\).
Construct a consistent line metric $d$ by placing the voters preferring \(i\) at \(i\), and placing the remaining voters at the midpoint between \(i\) and \(j\), breaking their distance ties in favor of \(j\).
Then $\frac{\operatorname{SC}(j,d)}{\operatorname{SC}(i,d)}=1+2\sqrt{2}$.
For finite electorates, choose rational weights for which $\frac{w_{ij}}{w_{ji}}$ approaches $\sqrt{2}$ from below.
The resulting distortion approaches $1+2\sqrt{2}$.

\section{A Five- and Six-Candidate Improvement}\label{sec:local-bounds}

The tight distortion bound of $3$ for at most four candidates is already known \cite{anshelevich2018approximating} (see also \cite[Theorem~A.2]{charikar2025metric}), so we begin with the first unresolved case of five candidates.
The existence proof in Theorem~\ref{thm:main-upper-bound} applies to cycles of arbitrary length.
For five or six candidates, we can examine all possible bad cycles and prove a stronger bound with computer assistance.
Let $\rho$ be the unique real root greater than one satisfying $\rho^5=\rho+1$, with $\rho\approx1.1673039$.

\begin{theorem}\label{thm:six-candidate-bound}
For each $m\in\{5,6\}$, the $\rho$-Path-Unblanketed Set rule is a deterministic C2 rule with distortion at most $1+2\rho\approx 3.3346$.
\end{theorem}

\subsection{Reduction to Linear Programs}\label{subsec:local-reduction}

The following lemma gives inequalities that must hold when a proxy is either inadmissible or unreachable.

\begin{lemma}\label{lem:projected-closed-sets}
Suppose that the bad-pair relation contains a directed cycle on a vertex set $Q\subseteq C$.
For every bad pair $(i,j)$ on this cycle and every $k\in Q\setminus\{j\}$, at least one of the following statements holds:
\begin{enumerate}
    \item $k\ne i$ and $w_{ij}>\lambda w_{ki}$, so $k$ is not admissible for $(i,j)$.
    \item There is a set $S_{i,j,k}\subseteq Q$ such that $j\in S_{i,j,k}$, $k\notin S_{i,j,k}$, and $w_{bj}>\lambda w_{ak}$ for every $a\in S_{i,j,k}$ and $b\in Q\setminus S_{i,j,k}$.
\end{enumerate}
\end{lemma}

\begin{proof}
Fix $(i,j)$ and $k$.
If $k$ is not admissible, then the definition of admissibility gives $k\ne i$ and $w_{ij}>\lambda w_{ki}$, so the first alternative holds.
Suppose instead that $k$ is admissible.
Because $(i,j)$ is bad, $k$ is not reachable from $j$ in $G_\lambda(j,k)$.
Let $\widehat S$ be the set of all candidates reachable from $j$ in this graph, and set $S_{i,j,k}=\widehat S\cap Q$.
Then $j\in S_{i,j,k}$ and $k\notin S_{i,j,k}$.
If some $a\in S_{i,j,k}$ and $b\in Q\setminus S_{i,j,k}$ satisfied $w_{bj}\le\lambda w_{ak}$, the graph would contain the arc $a\to b$.
A path from $j$ to $a$ followed by this arc would make $b$ reachable.
Since $b\in Q$, this would imply $b\in S_{i,j,k}$, a contradiction.
Thus, the second alternative holds.
\end{proof}

We now set $\lambda=\rho$.
Fix a cycle length $q$ and identify its vertex set with $Q_q=\{0,1,\ldots,q-1\}$.
After cyclic relabeling, the bad edges are $(q-1,0),(0,1),\ldots,(q-2,q-1)$.
For every bad edge $(i,j)$ and every proxy $k\in Q_q\setminus\{j\}$, Lemma~\ref{lem:projected-closed-sets} supplies at least one of the following alternatives.
In the second alternative, the set is closed within the cycle: no arc of $G_\rho(j,k)$ goes from it to a cycle vertex outside it.

\begin{enumerate}
    \item If $k\ne i$, the \emph{nonadmissibility alternative} consists of the single comparison
    \begin{align}
        w_{ij}>\rho w_{ki}.
        \label{eq:appendix-nonadmissible}
    \end{align}

    \item For every set $S\subseteq Q_q$ with $j\in S$ and $k\notin S$, the \emph{closed-set alternative indexed by $S$} consists of all the comparisons
    \begin{align}
        w_{bj}>\rho w_{ak}
        \qquad
        (a\in S,\ b\in Q_q\setminus S).
        \label{eq:appendix-closed-set}
    \end{align}
\end{enumerate}

There are $2^{q-2}$ possible closed sets because membership is fixed for $j$ and $k$, while each of the remaining $q-2$ candidates may be placed on either side.
A \emph{witness pattern} chooses one available alternative for every triple $(i,j,k)$ consisting of a bad edge and a proxy $k\ne j$.
Choosing a closed-set alternative includes all the inequalities in \eqref{eq:appendix-closed-set}.

Every bad cycle gives a witness pattern: for each proxy, choose the nonadmissibility alternative if it is not admissible, and otherwise choose the set constructed in Lemma~\ref{lem:projected-closed-sets}.
The actual weights then satisfy all comparisons in the resulting witness pattern.
Therefore, to rule out bad cycles, it is enough to prove that every witness pattern is infeasible.

Before eliminating redundant alternatives as described in Section~\ref{subsec:common-margin}, the number of patterns is at most
\begin{align}
    \left[
        2^{q-2}\bigl(1+2^{q-2}\bigr)^{q-2}
    \right]^q.
    \label{eq:branch-count}
\end{align}
Indeed, on each bad edge the proxy $k=i$ has only the $2^{q-2}$ closed-set choices, whereas each of the other $q-2$ proxies has one nonadmissibility choice and $2^{q-2}$ closed-set choices.

Retain the $\binom q2$ independent variables $w_{ab}$ with $a<b$ and substitute $w_{ba}=1-w_{ab}$ whenever a comparison uses the reverse orientation.
Each selected comparison is then a strict linear inequality in these variables, each of which lies in $[0,1]$.
Let $\mathcal E$ be the comparisons in a witness pattern.
We replace the strict inequalities by inequalities with a common slack $\delta$ and solve
\begin{align}
    \max\quad &\delta
    \notag\\
    \text{subject to}\quad
    &w_{xy}-\rho w_{uv}\ge\delta
    &&\text{for every }((u,v),(x,y))\in\mathcal E,
    \notag\\
    &0\le w_{ab}\le1
    &&\text{for every }a<b.
    \label{eq:local-margin-lp}
\end{align}

The strict comparisons in $\mathcal E$ are simultaneously feasible if and only if the optimum of \eqref{eq:local-margin-lp} is positive.
A feasible solution with $\delta>0$ satisfies every strict inequality.
Conversely, if all comparisons hold strictly, their minimum slack is positive and gives a feasible value of $\delta$.
Thus, the witness pattern is impossible if the optimum is at most zero.

\subsection{Excluding Infeasible Cases}\label{subsec:common-margin}

We enumerate witness patterns by choosing an alternative for one pair $(i,j)$ and proxy $k$ at a time.
These choices form a branching tree, whose nodes correspond to partial choices and whose leaves correspond to complete witness patterns.
At each node, we apply \eqref{eq:local-margin-lp} to the comparisons selected so far.
If its optimum is at most zero, no extension can satisfy all the strict inequalities, so we prune the entire branch.

We also remove redundant alternatives before branching.
Suppose two alternatives for the same triple have comparison sets $\mathcal E_1\subseteq\mathcal E_2$.
The second alternative is stronger: every assignment satisfying it also satisfies the first.
We retain only the first alternative, since proving it infeasible also excludes the second.

To prove that an LP has optimum at most zero, we take a nonnegative linear combination of its constraints.
After substituting $w_{ba}=1-w_{ab}$, write a selected comparison as $g_e(w)-\delta\ge0$, where $g_e$ is affine in the independent weights.
The bounds are written as $w_{ab}\ge0$ and $1-w_{ab}\ge0$.
A pruning certificate is a collection of nonnegative multipliers $\alpha_e,\beta_{ab},\gamma_{ab}$ for which the identity
\begin{align}
    &\sum_{e\in\mathcal E}\alpha_e\bigl(g_e(w)-\delta\bigr)
      +\sum_{a<b}\beta_{ab}w_{ab}
      +\sum_{a<b}\gamma_{ab}(1-w_{ab})
      =-c\delta-\eta
    \label{eq:exact-certificate}
\end{align}
holds, where $c>0$ and $\eta\ge0$.
At any feasible point, every term on the left is nonnegative.
Hence $-c\delta-\eta\ge0$, which implies $\delta\le-\frac{\eta}{c}\le0$.
This rules out every witness pattern extending the current partial choice.

For illustration, consider $q=3$, with bad edges $(2,0),(0,1),(1,2)$, and write
\begin{align*}
    x=w_{01},\qquad y=w_{12},\qquad z=w_{20}.
\end{align*}
Then $w_{10}=1-x$, $w_{21}=1-y$, and $w_{02}=1-z$.
For each bad edge, the proxy $k=i$ has two closed-set alternatives, while the third candidate has one nonadmissibility alternative and two closed-set alternatives.
Formula~\eqref{eq:branch-count} therefore gives $6^3=216$ patterns before pruning.

First, choose the nonadmissibility alternative for the third candidate on each bad edge.
The selected comparisons include
\begin{align*}
    z-\rho y&\ge\delta,
    &x-\rho z&\ge\delta,
    &y-\rho x&\ge\delta.
\end{align*}
Adding these inequalities and then adding $(\rho-1)$ times each valid bound $x\ge0$, $y\ge0$, and $z\ge0$ gives the identity
\begin{align*}
    &(z-\rho y-\delta)
      +(x-\rho z-\delta)
      +(y-\rho x-\delta)
      +(\rho-1)(x+y+z)
      =-3\delta.
\end{align*}
All multipliers are nonnegative because $\rho>1$.
This gives \eqref{eq:exact-certificate} with $c=3$ and $\eta=0$, so $\delta\le0$.
The branch can be pruned without considering the alternatives for the proxies $k=i$.

For a second example, consider the pattern that chooses the smallest closed set $S=\{j\}$ for every proxy.
Two comparisons already suffice.
From the triple $(i,j,k)=(2,0,1)$ we obtain
\begin{align*}
    (1-x)-\rho x\ge\delta,
\end{align*}
and from $(i,j,k)=(0,1,0)$ we obtain
\begin{align*}
    x-\rho(1-x)\ge\delta.
\end{align*}
Adding them gives
\begin{align*}
    (1-x-\rho x-\delta)
    +(x-\rho(1-x)-\delta)
    =-2\delta-(\rho-1).
\end{align*}
This gives a certificate with $c=2$ and $\eta=\rho-1>0$, proving $\delta\le\frac{1-\rho}{2}<0$.
These two comparisons therefore suffice to exclude the pattern.

\subsection{Computer-Assisted Proof}\label{subsec:exact-verification}

The computation uses exact arithmetic in $\mathbb Q(\rho)$ to check the multipliers and identity in \eqref{eq:exact-certificate}.
Each entry is represented as
\begin{align*}
    a_0+a_1\rho+a_2\rho^2+a_3\rho^3+a_4\rho^4,
    \qquad a_0,\ldots,a_4\in\mathbb Q.
\end{align*}
Products are reduced using $\rho^5=\rho+1$, so, for example, $\rho^6=\rho^2+\rho$.
Equality is checked coefficientwise after this reduction.
Signs are checked using rational interval arithmetic and the isolating interval
\begin{align}
    \frac{1167303978}{10^9}
    <\rho<
    \frac{1167303979}{10^9}.
    \label{eq:rho-isolating-interval}
\end{align}
For each nonzero quantity checked in the computation, interval arithmetic using these bounds gives an interval entirely above or below zero.
A quantity is zero if all coefficients in its reduced representation are zero.

A floating-point LP solver first identifies constraints to use in the linear combination.
The corresponding multipliers are then reconstructed in $\mathbb Q(\rho)$.
We check that they are nonnegative and satisfy \eqref{eq:exact-certificate} before pruning the branch.

\begin{lemma}\label{lem:exact-six-cycle-exclusion}
For $\lambda=\rho$, the alternatives in Lemma~\ref{lem:projected-closed-sets} cannot hold simultaneously over all $i,j,k$ on a bad cycle of any length $q\in\{2,3,4,5,6\}$.
\end{lemma}

\begin{proof}
For each $q\in\{2,3,4,5,6\}$, the verifier traverses the finite branching tree of witness patterns described above.
At each node, it checks a certificate of the form \eqref{eq:exact-certificate} and prunes the branch if the check succeeds.
Otherwise, it branches on the alternatives for the next pair $(i,j)$ and proxy $k$, after removing redundant alternatives as above.
For all five cycle lengths, the computation excludes every complete witness pattern, either directly or by pruning an earlier partial choice.
Hence every witness pattern has LP optimum at most zero.
By the equivalence proved in Section~\ref{subsec:local-reduction}, no witness pattern has a solution satisfying all its strict inequalities.
Every bad cycle would give such a witness pattern, so the stated alternatives cannot hold simultaneously.
\end{proof}

\begin{proof}[Proof of Theorem~\ref{thm:six-candidate-bound}]
Fix $m\in\{5,6\}$ and suppose for contradiction that a weighted tournament on $m$ candidates has no $\rho$-path-unblanketed candidate.
Lemma~\ref{lem:bad-cycle} then gives a directed cycle of bad pairs of some length $q\in\{2,\ldots,m\}$.
Applying Lemma~\ref{lem:projected-closed-sets} to this cycle produces one of its two alternatives for every bad pair and every proxy on the cycle.
This contradicts Lemma~\ref{lem:exact-six-cycle-exclusion}, so a $\rho$-path-unblanketed candidate $j$ always exists.
Use the Path-Unblanketed Set rule with parameter $\lambda=\rho$ and the same fixed tie-breaking order as before.
The preceding existence argument makes the rule well defined, and its tests use only the weighted tournament.
Fix a consistent metric $d$, let $i$ be a socially optimal candidate, and let $j$ be the selected candidate.
If $i=j$, the distortion bound is immediate.
Otherwise, $j$ has a reachable $\rho$-admissible proxy for $(i,j)$, so Lemma~\ref{lem:path-certificate} gives $\mathrm{SC}(j,d)\le(1+2\rho)\mathrm{SC}(i,d)$, proving the theorem.
\end{proof}

\section{Lower Bound}
\label{sec:lower-bound}

This section gives a computer-assisted lower bound for deterministic C2 rules.
We use the biased-metric construction of Charikar and Ramakrishnan~\cite{charikar2022metric}, following its use in Charikar et al.~\cite{charikar2025metric}, and find suitable ranking distributions by linear programming.
Let $\tau$ be the unique real root greater than one of
\begin{align}
    3\tau^3+\tau^2-\tau-4=0.
    \label{eq:lower-tau-polynomial}
\end{align}
Numerically, $\tau\approx1.0914$.

\begin{theorem}
\label{thm:exact-c2-lower-bound}
For every $m\ge17$,
\begin{align}
    D_m^{\mathrm{C2}}
    \ge 1+2\tau
    \approx3.1828.
    \label{eq:exact-c2-lower-bound}
\end{align}
\end{theorem}

\subsection{Lower-Bound Framework}

We construct several distributions over rankings that induce the same weighted tournament.
For each candidate $j$, one of these distributions admits a consistent metric in which $j$ has a large social cost relative to an optimal candidate.
A deterministic C2 rule selects the same winner for all these distributions, so the instance corresponding to that winner gives the lower bound.

Write $\mathcal L(C)$ for the set of strict rankings of $C$.
For every possible output $j\in C$, choose a distribution $p^{(j)}$ over $\mathcal L(C)$, a candidate $i_j$, and a nonnegative vector $x^{(j)}\in\mathbb R_{\ge0}^C$ with $x^{(j)}_{i_j}=0$.
For a ranking $\pi$, define
\begin{align}
    A_j(\pi)
    &=\min_{k:\,j\succeq_\pi k}x^{(j)}_k,
    &
    B_j(\pi)
    &=\max_{a\succeq_\pi b}
      \bigl(x^{(j)}_a-x^{(j)}_b\bigr).
    \label{eq:lower-ab-statistics}
\end{align}
The weak comparison in the definition of $B_j$ allows $a=b$, so $B_j(\pi)\ge0$.
We normalize the total voter mass of each distribution to one throughout this section.
Multiplying all social costs by the size of a finite electorate does not affect any ratio.

The biased-metric construction~\cite{charikar2022metric,charikar2025metric} associates a consistent pseudo-metric with $p^{(j)}$ and $x^{(j)}$ under which $i_j$ is socially optimal and
\begin{align}
    \mathrm{SC}(j)-\mathrm{SC}(i_j)
    &=\mathbb E_{p^{(j)}}[A_j],
    &
    2\mathrm{SC}(i_j)
    &=\mathbb E_{p^{(j)}}[B_j].
    \label{eq:lower-biased-metric-identities}
\end{align}
Consequently, whenever $\mathbb E[B_j]>0$,
\begin{align}
    \frac{\mathrm{SC}(j)}{\mathrm{SC}(i_j)}
    =1+2\frac{\mathbb E[A_j]}{\mathbb E[B_j]}.
    \label{eq:lower-biased-metric-ratio}
\end{align}

\begin{lemma}
\label{lem:common-tournament-lower-certificate}
Suppose that all distributions $p^{(j)}$ induce the same weighted tournament and that, for some $\lambda\ge1$ and every $j\in C$,
\begin{align}
    \mathbb E_{p^{(j)}}[A_j]
    \ge\lambda\mathbb E_{p^{(j)}}[B_j]>0.
    \label{eq:lower-certificate-inequality}
\end{align}
Then every deterministic C2 rule has distortion at least $1+2\lambda$ on this candidate set.
\end{lemma}

\begin{proof}
Let $W$ be the common weighted tournament and consider an arbitrary deterministic C2 rule.
On input $W$, the rule returns some fixed candidate $j$.
Consider the ranking distribution $p^{(j)}$ and its associated biased metric.
Since $p^{(j)}$ induces $W$, the rule selects $j$ on this instance.
Equations \eqref{eq:lower-biased-metric-ratio} and \eqref{eq:lower-certificate-inequality} give distortion at least $1+2\lambda$.
\end{proof}

\subsection{A Construction with 17 Candidates}

Identify the candidates with $\mathbb Z_{17}$ and set $i_j=j-1$ modulo 17.
For each candidate $j$, the entries of $x^{(j)}$ belong to $\{0,\frac12,1\}$.
For an integer $r_j$, define
\begin{align}
    x^{(j)}_{j-1}&=0,
    &
    x^{(j)}_k&=1
        \quad\text{for }k\in\{j,j+1,\ldots,j+r_j\},
    \label{eq:lower-interval-metric}
\end{align}
with indices taken modulo 17, and set every remaining coordinate to $\frac12$.
The construction uses
\begin{align}
    (r_0,r_1,\ldots,r_{16})
    =(0,1,15,15,15,15,15,12,15,8,7,6,5,4,3,2,1).
    \label{eq:lower-interval-lengths}
\end{align}
When $r_j=15$, every candidate other than $i_j$ has coordinate one.
For smaller $r_j$, the candidates outside the displayed interval and distinct from $i_j$ have coordinate $\frac12$.

For each candidate $j$, we use a distribution $p^{(j)}$ over rankings.
Table~\ref{tab:lower-certificate-summary} lists the support sizes of the 17 distributions, which sum to 1810.
Each positive ranking weight belongs to $\mathbb Q(\tau)$ and has the form
\begin{align}
    p^{(j)}_\pi
    =\frac{c_0+c_1\tau+c_2\tau^2}{Q_j},
    \qquad c_0,c_1,c_2,Q_j\in\mathbb Z.
    \label{eq:lower-algebraic-weight}
\end{align}
The table also gives $\mathbb E[A_j]/\mathbb E[B_j]$ for each candidate.
The entries marked $\tau$ satisfy $\mathbb E[A_j]=\tau\mathbb E[B_j]$, and the other five ratios are strictly larger than $\tau$.
\begin{table}[t]
    \centering
    \caption{The 17-candidate construction, with comparison candidate $i_j=j-1$ modulo 17.}
    \label{tab:lower-certificate-summary}
    \begin{tabular}{c@{\quad}c@{\quad}c@{\quad}c}
        \toprule
        candidate $j$ & $r_j$ & support size &
        $\mathbb E[A_j]/\mathbb E[B_j]$ \\
        \midrule
        $0$  & $0$  & $94$  & $\tau$ \\
        $1$  & $1$  & $92$  & $\tau$ \\
        $2$  & $15$ & $112$ & $\tau$ \\
        $3$  & $15$ & $125$ & $\tau$ \\
        $4$  & $15$ & $124$ & $\tau$ \\
        $5$  & $15$ & $122$ & $\tau$ \\
        $6$  & $15$ & $124$ & $\tau$ \\
        $7$  & $12$ & $105$ & $\tau$ \\
        $8$  & $15$ & $123$ & $1.0914193694\ldots$ \\
        $9$  & $8$  & $115$ & $\tau$ \\
        $10$ & $7$  & $107$ & $1.0914169098\ldots$ \\
        $11$ & $6$  & $100$ & $\tau$ \\
        $12$ & $5$  & $94$  & $\tau$ \\
        $13$ & $4$  & $93$  & $\tau$ \\
        $14$ & $3$  & $89$  & $1.0914658053\ldots$ \\
        $15$ & $2$  & $96$  & $1.0914533585\ldots$ \\
        $16$ & $1$  & $95$  & $1.0920732282\ldots$ \\
        \bottomrule
    \end{tabular}
\end{table}

\subsection{Computing the Ranking Distributions}

We now compute the ranking weights on the supports found by linear programming.
For each $j$, form the $137$-by-$s_j$ zero-one matrix $M_j$ whose column for ranking $\pi$ is
\begin{align*}
    \bigl(1,(\mathbf 1[a\succ_\pi b])_{0\le a<b\le16}\bigr),
\end{align*}
where $s_j$ is the support size of $p^{(j)}$, given in Table~\ref{tab:lower-certificate-summary}.
Multiplying $M_j$ by the vector of ranking weights gives their total weight followed by the pairwise frequencies.
Thus, the common vector for all 17 distributions must lie in $\bigcap_{j=0}^{16}\operatorname{col}(M_j)$, with first coordinate one.
Computing the intersection over $\mathbb Q$ gives dimension 13.
Choose a rational basis for this intersection and write its vectors as $y=Tt$, where $T\in\mathbb Q^{137\times13}$ and $t\in\mathbb R^{13}$.
Solving the linear equations for the ranking weights expresses them as linear functions of $t$ with rational coefficients.

We impose normalization and the equalities $\mathbb E[A_j]=\tau\mathbb E[B_j]$ for
\begin{align}
    j\in\{0,1,2,3,4,5,6,7,9,11,12,13\}.
    \label{eq:lower-active-attacks}
\end{align}
The resulting system over $\mathbb Q(\tau)$ has rank 12, leaving one free parameter.
Setting the free parameter to $s=\frac{16407}{25000}$ gives the weights in \eqref{eq:lower-algebraic-weight}, all of which are strictly positive.

We use exact arithmetic to check that the resulting distributions satisfy the following conditions:
\begin{enumerate}
    \item every order in the support is a permutation of the candidates and every weight is strictly positive;
    \item the weights in each of the 17 distributions sum to one;
    \item all 17 distributions give the same values for the 136 pairwise frequencies $w_{ab}$ with $a<b$; and
    \item for every candidate $j$, the quantities in \eqref{eq:lower-ab-statistics} satisfy $\mathbb E[A_j]-\tau\mathbb E[B_j]\ge0$.
\end{enumerate}
Arithmetic is performed in the basis $1,\tau,\tau^2$, with higher powers reduced using
\begin{align*}
    \tau^3=\frac{-\tau^2+\tau+4}{3}.
\end{align*}
We determine signs using rational interval arithmetic.
The initial isolating interval is
\begin{align}
    \frac{1091414260213383}{10^{15}}
    <\tau<
    \frac{1091414260213384}{10^{15}}.
    \label{eq:lower-tau-isolating-interval}
\end{align}
The polynomial in \eqref{eq:lower-tau-polynomial} changes sign across this interval, and its derivative $9x^2+2x-1$ is positive for $x\ge1$, proving that the interval contains the unique root greater than one.
For each nonzero quantity, we refine the interval by bisection until its sign is determined.
Zero is recognized by reducing to the basis $1,\tau,\tau^2$ and checking that all coefficients vanish.

\subsection{Proof of the Lower Bound}

\begin{proof}[Proof of Theorem~\ref{thm:exact-c2-lower-bound}]
We first obtain finite profiles with equally weighted voters from the distributions above.
Let $H$ be the affine space of ranking weights on the chosen supports that satisfy normalization and give identical pairwise frequencies across all 17 distributions.
These are rational linear equations, so rational points are dense in $H$.
The weights constructed above give a point of $H$ at which every ranking weight is strictly positive.

Fix any $\lambda<\tau$.
At this point,
\begin{align*}
    \mathbb E[A_j]-\lambda\mathbb E[B_j]
    =\bigl(\mathbb E[A_j]-\tau\mathbb E[B_j]\bigr)
      +(\tau-\lambda)\mathbb E[B_j]>0
\end{align*}
for every $j$.
Positivity of the weights and these strict inequalities are preserved in a sufficiently small relative neighborhood in $H$.
This neighborhood contains a rational point.
Clearing a common denominator then turns all 17 rational distributions into finite electorates of the same size, still inducing the same weighted tournament and satisfying $\mathbb E[A_j]>\lambda\mathbb E[B_j]$ for every $j$.
Lemma~\ref{lem:common-tournament-lower-certificate} gives a finite-profile lower bound greater than $1+2\lambda$.
Taking the supremum as $\lambda\uparrow\tau$ proves Theorem~\ref{thm:exact-c2-lower-bound} for $m=17$.

It remains to extend the result to larger candidate sets.
Given such a construction on $m$ candidates, choose a candidate $c$ and add a clone $c'$.
In every supported ranking, replace $c$ by the adjacent block $(c,c')$ or $(c',c)$, each with half the original weight.
The original pairwise frequencies are unchanged, $c$ and $c'$ have identical pairwise frequencies against all other candidates, and $w_{cc'}=\frac12$.
Extend every metric by setting $x_{c'}=x_c$.
Because the clones are adjacent and have equal coordinates, neither $A_j$ nor $B_j$ changes for any original candidate $j$.
For output $c'$, exchange the labels of $c$ and $c'$ in the instance constructed for $c$.
This exchange leaves the weighted tournament unchanged.
The construction therefore gives the same lower bound with $m+1$ candidates.
Repeating the construction proves \eqref{eq:exact-c2-lower-bound} for every $m\ge17$.
\end{proof}

\section{Conclusion}\label{sec:conclusion}

We introduced the Path-Unblanketed Set rule, a polynomial-time deterministic C2 rule with distortion at most $1+2\sqrt{2}\approx3.8284$.
The rule extends the Unblanketed Set of Charikar et al.~\cite{charikar2025metric} by allowing comparisons along paths.
For five or six candidates, we proved a stronger upper bound of $3.3346$ with computer assistance.
We also improved the general lower bound from $3.1128$ to $3.1828$ using a construction with 17 candidates.

Determining the optimal distortion for deterministic C2 rules remains open.
One direction is to strengthen the social-cost comparison of Charikar et al.~\cite{charikar2025metric}, perhaps by using several proxies together.
Another is to use the additional restrictions on pairwise frequencies that arise from a common distribution over rankings, beyond the identities $w_{ab}+w_{ba}=1$ used in our existence proof.
For small candidate sets, it would also be interesting to extend our upper-bound analysis to seven or more candidates and determine the optimal distortion for each number of candidates.

\section*{AI Use Statement}
We used OpenAI GPT-based tools, including ChatGPT and Codex, to assist in developing and checking several intermediate lemmas and their proofs.
For the computer-assisted bounded-candidate upper bounds and lower bounds, these tools played a primary role in designing the search and certificate procedures and implementing the verification code.
The authors independently validated all AI-assisted outputs and take full responsibility for the paper's content.

\bibliographystyle{plain}
\bibliography{mybibfile}

\end{document}